\documentclass[a4paper, 11pt]{article}
\usepackage{ascmac}
\usepackage{mathrsfs}
\usepackage{latexsym}
\usepackage{amsmath}
\usepackage{amssymb}
\usepackage{amsthm}

\usepackage[hdivide={30mm,,30mm}, vdivide={20mm,,25mm}, nohead]{geometry}

 \makeatletter
    
    \@addtoreset{equation}{section}
  \makeatother

\newtheorem{Def}{Definition}[section]
\newtheorem{Thm}{Theorem}[section]
\newtheorem{Lem}{Lemma}[section]
\newtheorem{Prop}{Proposition}[section]

\newtheorem{Exam}{Example}[section]
\newtheorem{Ass}{Assumption}[section]

\begin{document}
\title{A weak limit theorem for a class of 
long range type \\ quantum walks in 1d
}
\author{Kazuyuki Wada\\ 
 {National Institute of Technology, Hachinohe college.}\\\ Hachinohe, 039-1192, Japan.\vspace{3mm}\\ E-mail address: wada-g@hachinohe.kosen-ac.jp  
}
\date{}
\maketitle
\begin{abstract}
We derive the weak limit theorem for a class of long range type quantum walks. To do it, we analyze spectral properties of a time evolution operator and prove that modified wave operators exist and are complete.
\end{abstract}

\section{Introduction} 

Quantum walks have been introduced as a quantum counter part of classical random walks [1,5]. It is known that quantum walks have remarkable properties which are not seen in classical random walks. One of these properties appears in a \lq\lq weak limit theorem". In \cite{K}, Konno firstly derived the limit distribution of quantum walks. He also revealed that the shape of limit distribution is quite different from the normal distribution. 

Here to explain some results related to weak limit theorem, we briefly introduce a mathematical framework of quantum walks. The Hilbert space is 
\begin{equation*}
\mathcal{H}:=l^{2}(\mathbb{Z};\mathbb{C}^{2})=\{\Psi:\mathbb{Z}\rightarrow \mathbb{C}^{2}|\displaystyle\sum_{x\in\mathbb{Z}}\|\Psi(x)\|^{2}_{\mathbb{C}^{2}}<\infty\},
\end{equation*}
and time evolution operator is $U:=SC$ where 
\begin{equation*}
(S\Psi)(x)=\begin{bmatrix}\Psi^{(1)}(x+1) \\ \Psi^{(2)}(x-1)\end{bmatrix}, \hspace{2mm}(C\Psi)(x)=C(x)\Psi(x),\hspace{2mm}\Psi\in\mathcal{H},\hspace{3mm}x\in\mathbb{Z},
\end{equation*}
and $\{C(x)\}_{x\in\mathbb{Z}}\subset U(2)$. Let $\Psi_{0}\in\mathcal{H}$ ($\|\Psi\|=1$) be an initial state of a quantum walker. Then the quantum state after time $t\in\mathbb{Z}$ is given by $U^{t}\Psi_{0}$. 

For $\Psi_{0}\in\mathcal{H}$ with $\|\Psi_{0}\|=1$ and $t\in\mathbb{Z}$, $X_{t}$ be a $\mathbb{Z}$-valued random variable whose probability distribution is given by $\mathbb{P}(X_{t}=x)=\|(U^{t}\Psi_{0})(x)\|^{2}_{\mathbb{C}^{2}}$. Our interest is to find the random variable $V$ such that $X_{t}$ with a suitable scaling converges to $V$ as $t\rightarrow\infty$.
Konno considered space-homogeneous quantum walks in one dimension. It means that $C(x)=C_{0}$ ($x\in\mathbb{Z}$) for some $C_{0}\in U(2)$. He assume that the initial state $\Psi_{0}\in\mathcal{H}$ has a form of 
\begin{equation*}
\Psi(x)=\begin{cases} \begin{bmatrix}\alpha \\ \beta\end{bmatrix}\hspace{5mm}x=0, \\ 
\begin{bmatrix} 0 \\ 0\end{bmatrix}\hspace{5mm}\text{otherwise},
\end{cases}\hspace{3mm}(|\alpha|^{2}+|\beta|^{2}=1).
\end{equation*}
Then he showed that the existence of $\mathbb{R}$-valued random variable $V$ such that $X_{t}/t\rightarrow V$ as $t\rightarrow\infty$ in a weak sense through combinatrical arguments [8]. After that, Grimmett et al. gave a simple proof for Konno's result, extended to $d$-dimensional space-homogeneous quantum walks and removed the assumption related to initial states [6]. Their proof is based on an application of the discrete Fourier transform. A crucial contribution is to find the self-adjoint operator $\hat{V}_{0}$ which induces the random variable $V$. $\hat{V}_{0}$ is called an \lq\lq asymptotic velocity operator". To find the limit distribution of $X_{t}/t$ as $t\rightarrow\infty$, it suffices to find a suitable asymptotic velocity operator. Recently, a nonlinear quantum walk is considered in [9].

If we allow a coin operator $C$ to be depend on $x\in\mathbb{Z}$, it becomes difficult to obtain the weak limit theorem since the discrete Fourier transform does not work. To overcome this difficulty, Suzuki introduced the idea of spectral scattering theory for quantum walks. Here, we introduce the notion of short range type and long range type conditions:
\begin{Def}\normalfont
A coin operator $C$ satisfy a short (resp. long) range type condition if there exists $C_{0}\in U(2)$, $\kappa>0$, and $\gamma>1$ (resp. $ 1\ge\gamma>0$) such that
\begin{equation*}
\|C(x)-C_{0}\|_{\mathcal{B}(\mathbb{C}^{2})}\le \kappa (1+|x|)^{-\gamma},\hspace{3mm} x\in\mathbb{Z},
\end{equation*}
where $\|\cdot\|_{\mathcal{B}(\mathbb{C}^{2})}$ is the operator norm on $\mathbb{C}^{2}$.
\end{Def}
We assume that $C$ satisfies the short range type condition. We set $U_{0}:=SC_{0}$. Then the following wave operator
\begin{equation*}
W_{\pm}:=\text{s-}\displaystyle\lim_{t\rightarrow\pm\infty}U^{-t}U_{0}^{t}\Pi_{\text{ac}}(U_{0})
\end{equation*}
exist and are complete (i.e. \text{Ran}$W_{\pm}=\mathcal{H}_{\text{ac}}(U)$). Moreover, we can show that the absence of singular continuous spectrum of $U$ by Mourre theory [3,12]. We denote the asymptotic velocity operator of $U_{0}$ by $\hat{V}_{0}$. Suzuki showed that the limit distribution of $X_{t}/t$ as $t\rightarrow\infty$ is derived from a sum of the orthogonal projection onto the set of eigenvectors of $U$ and the spectral measure of $W_{+}^{\ast}\hat{V}_{0}W_{+}$. 

On the other hand, in the long range type condition, wave operators do not exist in general [16]. Thus it is not trivial problem how to get the limit distribution of $X_{t}/t$. 

In scattering theory for quantum mechanics, it is known that we have to introduce modified wave operators instead of wave operators. There are lots of results related to long range scattering theory [4, 11]. To introduce modified wave operators, it is important to introduce a suitable \lq\lq modifier" induced by the Hamiltonian. However, it is difficult to introduce a modifier in a context of quantum walks straightforwardly since the Hamiltonian corresponds to $U=SC$ is unknown in general. The generator of quantum walks is studied in [14].

In this paper, we derive the weak limit theorem for a class of $U$ whose coin operator satisfies the long range type condition. We assume that a coin operator $C$ has a form of 
\begin{equation*}
C(x)=\begin{bmatrix}e^{-i\xi(x)} & 0 \\ 0 & e^{i\xi(x)}\end{bmatrix}C_{0},
\end{equation*}
for some $\xi:\mathbb{Z}\rightarrow\mathbb{R}$ and $C_{0}\in U(2)$. As far as we know, this is the first result related to long range type quantum walks. To derive the weak limit theorem, it is important to show the absence of singular continuous spectrum of $U$ and existence of modified wave operators. We apply commutator theory for unitary operators under two Hilbert space settings established by Richard et al. [13] and Kato-Rosenblum type theorem established by Suzuki [15].

Contents of this paper is as follows. In section 2, we give a definition of a model in quantum walks and some fundamental properties are explained. In section 3, some facts in the commutator theory is introduced. In section 4, we show the absence of singular continuous spectrum of $U$ by applying the commutator theory explained in section 3. In section 5, we derive the weak limit theorem which is a main result in this paper.

\section{Definition of a model}
In this section we review some notations and fundamental results for quantum walks. The Hilbert space is given by
\begin{equation}
\mathcal{H}:=l^{2}(\mathbb{Z};\mathbb{C}^{2})=\Big\{\Psi:\mathbb{Z}\rightarrow\mathbb{C}^{2}\Big| \displaystyle\sum_{x\in\mathbb{Z}}\|\Psi(x)\|^{2}_{\mathbb{C}^{2}}<\infty\Big\},
\end{equation}
where $\|\cdot\|_{\mathbb{C}^{2}}$ is the norm on $\mathbb{C}^{2}$. We denote its inner product and norm by $\langle\cdot, \cdot \rangle_{\mathcal{H}}$ (linear in the right vector)  and $\|\cdot\|_{\mathcal{H}}$, respectively. If there is no danger of confusion, then we omit the subscript $\mathcal{H}$ of them. We introduce the following dense subspace of $\mathcal{H}$:
\begin{equation}
\mathcal{H}_{\text{fin}}:=\{\Psi\in\mathcal{H}| \exists N\in\mathbb{N}\text{ such that }\Psi(x)=0\text{ for all }|x|\ge N\}.
\end{equation}
Next we introduce two unitary operators $U$ and $U_{0}$. For $\Psi\in\mathcal{H}$, the shift operator $S$ is defined by
\begin{equation}
(S\Psi)(x):=\begin{bmatrix}\Psi^{(1)}(x+1) \\ \Psi^{(2)}(x-1)\end{bmatrix}, \hspace{5mm}x\in\mathbb{Z}.
\end{equation}
Let $C_{0}$ be a $2\times2$ unitary matrix. We introduce the coin operator $C$ as follows:
\begin{equation}\label{coin}
(C\Psi)(x):=C(x)\Psi(x), \hspace{5mm}C(x):=\begin{bmatrix}e^{-i\xi(x)} & 0 \\ 0 & e^{i\xi(x)} \end{bmatrix}C_{0},\hspace{5mm}x\in\mathbb{Z},
\end{equation}
where $\xi$ is a real-valued function on $\mathbb{Z}$.
Throughout in this paper, we identify $C_{0}$ as a unitary operator on $\mathcal{H}$ such that  $(C_{0}\Psi)(x)=C_{0}\Psi(x)$, $x\in\mathbb{Z}$. We set $U:=SC$ and $U_{0}:=SC_{0}$.

Next, we recall spectral properties of $U_{0}=SC_{0}$. We denote the discrete Fourier transform which is unitary from $\mathcal{H}$ to $\mathcal{K}:=L^{2}([0, 2\pi), \text{d}k/2\pi;\mathbb{C}^{2})$ and
\begin{equation*}
(\mathcal{F}\phi)(k):=\hat{\phi}(k)=\displaystyle\sum_{x\in\mathbb{Z}}\phi(x)e^{-ikx},\hspace{2mm}k\in[0, 2\pi), \hspace{2mm}\phi\in\mathcal{H}_{\text{fin}}.
\end{equation*}
We set $\hat{U}_{0}:=\mathcal{F}U_{0}\mathcal{F}^{-1}$. It is seen that $\hat{U}_{0}$ is a $U(2)$-valued multiplication operator given by
\begin{equation*}
\hat{U}_{0}(k)=\begin{bmatrix}e^{ik} & 0 \\ 0 & e^{-ik}\end{bmatrix}C_{0},\hspace{5mm}k\in[0, 2\pi).
\end{equation*}
Note that $C_{0}$ has a form of
\begin{equation*}
C_{0}=\begin{bmatrix} ae^{i\alpha} & be^{i\beta} \\ -be^{-i\beta+i\delta} & ae^{-i\alpha+i\delta}\end{bmatrix},
\end{equation*}
where $a, b\in[0, 1]$ with $a^{2}+b^{2}=1$, $\alpha, \beta\in[0, 2\pi)$ and $e^{i\delta}$ ($\delta\in[0, 2\pi)$) is the determinant of $C_{0}$. We denote an eigenvalue and a correspond normalized eigenvector by $\lambda_{j}(k)$ and $u_{j}(k)$ $(j=1, 2)$, respectively.

Let $B$ be a unitary or self-adjoint operator on $\mathcal{H}$. The sets $\sigma(B)$, $\sigma_{\text{p}}(B)$, $\sigma_{\text{c}}(B)$, $\sigma_{\text{ess}}(B)$ and $\sigma_{\text{ac}}(B)$ are called spectrum, pure point spectrum, continuous spectrum, essential spectrum and absolutely continuous spectrum of $B$, respectively.

\begin{Prop}\normalfont [12, Lemma 4.1]
(1) If $a=0$, then
\begin{equation*}
\lambda_{1}(k)=ie^{i\delta/2},\hspace{3mm}\lambda_{2}(k)=-ie^{i\delta/2},
\end{equation*}
and
\begin{equation*}
\sigma(U_{0})=\sigma_{\text{p}}(U_{0})=\{ie^{i\delta}, -ie^{i\delta}\}.
\end{equation*}
(2) If $0<a<1$, then
\begin{equation*}
\lambda_{j}(k)=e^{i\delta/2}(\tau(k)+i(-1)^{j-1}\eta(k)), \hspace{3mm}j=1, 2,
\end{equation*}
where $\tau(k):=a\cos(k+\alpha-\delta/2)$ and $\eta(k):=\sqrt{1-\tau(k)^{2}}$. Moreover, it follows that
\begin{equation*}
\sigma(U_{0})=\sigma_{\text{c}}(U_{0})=\{e^{it}|t\in[\delta/2+\zeta, \pi+\delta/2-\zeta]\cup[\pi+\delta/2+\zeta, 2\pi+\delta/2-\zeta]\},
\end{equation*}
where $\zeta:=\arccos(a)$.
\\
(3) If $a=1$, then
\begin{equation*}
\lambda_{1}(k)=e^{i(k+\alpha)},\hspace{3mm}\lambda_{2}(k)=e^{-i(k+\alpha-\delta)},
\end{equation*}
and
\begin{equation*}
\sigma(U_{0})=\sigma_{\text{c}}(U_{0})=\mathbb{T}:=\{e^{it}|t\in[0, 2\pi)\}.
\end{equation*}
\end{Prop}

In what follows we assume that $a\in(0,1]$ ($C_{0}$ is not off-diagonal) to avoid a trivial case. 

For a given coin operator $C$ defined in (\ref{coin}), we introduce an important assumption:
\begin{Ass}\normalfont
Let $\xi:\mathbb{Z}\rightarrow\mathbb{R}$ be a function such that $\displaystyle\lim_{t\rightarrow\pm\infty}\xi(x)=0$. Then there exists $\theta:\mathbb{Z}\rightarrow\mathbb{R}$ such that
\begin{equation*}
\begin{cases}
\big|\xi(x)-\big\{\theta(x+1)-\theta(x)\big\}\big|\le \kappa (1+|x|)^{-1-\epsilon_{0}},
\\
\big|\xi(x)-\big\{\theta(x)-\theta(x-1)\big\}\big|\le \kappa (1+|x|)^{-1-\epsilon_{0}},
\end{cases}\hspace{5mm}x\in\mathbb{Z},
\end{equation*}
with some constants $\kappa>0$ and $\epsilon_{0}>0$.
\end{Ass}
\begin{Exam}\normalfont
If $\xi(x)=(1+|x|)^{-1}$, $x\in\mathbb{Z}$. Then we choose $\theta$ as follows:
\begin{equation*}
\theta(x)=\begin{cases} \log (1+x),\hspace{5mm}\text{if }x\ge0
\\
-\log(1-x),\hspace{5mm}\text{if }x<0
\end{cases}
\end{equation*}
Then, there exists $\kappa>0$ such that 
\begin{equation*}
\begin{cases}
\big|\xi(x)-\big\{\theta(x+1)-\theta(x)\big\}\big|\le \kappa (1+|x|)^{-2},
\\
\big|\xi(x)-\big\{\theta(x)-\theta(x-1)\big\}\big|\le \kappa (1+|x|)^{-2},
\end{cases}
\hspace{5mm}x\in\mathbb{Z}.
\end{equation*}
\end{Exam}
\begin{Exam}\normalfont
We can consider a generalization of Example 2.1. For $0<p<1$, we set $\xi(x)=(1+|x|)^{-p}$, $x\in\mathbb{Z}$. Then we choose $\theta$ as
\begin{equation*}
\theta(x)=\begin{cases} \displaystyle\frac{1}{1-p}(1+x)^{1-p}, \hspace{5mm}\text{if $x\ge 0$}, \\ -\displaystyle\frac{1}{1-p}(1-x)^{1-p},\hspace{5mm}\text{if $x<0$}\end{cases}
\end{equation*}
Then, there exists $\kappa>0$ such that
\begin{equation*}
\begin{cases}
\big|\xi(x)-\big\{\theta(x+1)-\theta(x)\big\}\big|\le \kappa (1+|x|)^{-1-p},
\\
\big|\xi(x)-\big\{\theta(x)-\theta(x-1)\big\}\big|\le \kappa (1+|x|)^{-1-p},
\end{cases}\hspace{5mm}x\in\mathbb{Z}
\end{equation*} 
\end{Exam}
In what follows, we assume the existence of $\theta$ which satisfy Assumption 2.1. We introduce the $U(2)$-valued multiplication operator $J$ as follows:
\begin{equation}\label{modifier}
(J\Psi)(x):=J(x)\Psi(x),\hspace{5mm}J(x)=\begin{bmatrix} e^{i\theta(x)} & 0 \\ 0 & e^{i\theta(x)}\end{bmatrix},\hspace{3mm}x\in\mathbb{Z}, \hspace{3mm}\Psi\in\mathcal{H}.
\end{equation}
It is obvious that $J$ is unitary on $\mathcal{H}$. We set $\tilde{U_{0}}:=JU_{0}J^{-1}$. Then we can express $\tilde{U}_{0}$ as  $\tilde{U_{0}}=S\tilde{C}_{0}$, where $\tilde{C}_{0}:=S^{-1}JSC_{0}J^{-1}$. 
\begin{Prop}\normalfont
$\tilde{C}_{0}$ is a $U(2)$-valued multiplication operator on $\mathbb{Z}$ such that
\begin{equation*}
\tilde{C}_{0}(x)=\begin{bmatrix} e^{-i\{\theta(x)-\theta(x-1)\}} & 0 \\ 0 & e^{i\{\theta(x+1)-\theta(x)\}}\end{bmatrix}C_{0},\hspace{3mm}x\in\mathbb{Z}.
\end{equation*}
\end{Prop}
\begin{proof}
Since $J^{-1}$ and $C_{0}$ commute, it suffices to consider the form of $(S^{-1}JSJ^{-1})(x)$. For any $\Psi\in\mathcal{H}$, it is seen that
\begin{equation*}
(JSJ^{-1}\Psi)(x)=\begin{bmatrix}e^{i\{\theta(x)-\theta(x+1)\}}\Psi^{(1)}(x+1)  \\ e^{i\{\theta(x)-\theta(x-1)\}}\Psi^{(2)}(x-1)\end{bmatrix}.
\end{equation*}
Moreover, it follows that
\begin{equation*}
\begin{aligned}
(S^{-1}JSJ^{-1}\Psi)(x)
&=\begin{bmatrix} (JSJ^{-1}\Psi)^{(1)}(x-1) \\ (JSJ^{-1}\Psi)^{(2)}(x+1)\end{bmatrix}
\\
&=\begin{bmatrix} e^{-i\{\theta(x)-\theta(x-1)\}}\Psi^{(1)}(x) \\ e^{i\{\theta(x+1)-\theta(x)\}}\Psi^{(2)}(x)\end{bmatrix}
\\
&=\begin{bmatrix}e^{-i\{\theta(x)-\theta(x-1)\}} & 0 \\ 0 & e^{i\{\theta(x+1)-\theta(x)\}}\end{bmatrix}\begin{bmatrix}\Psi^{(1)}(x) \\ \Psi^{(2)}(x)
\end{bmatrix}
\end{aligned}
\end{equation*}
Thus the desired result follows.
\end{proof}
By proposition 2.2 and $|e^{is}-1|\le |s|$ for $s\in\mathbb{R}$, we have the following proposition:
\begin{Prop}\normalfont
For any $x\in\mathbb{Z}$, it follows that
\begin{equation*}
\Big\|C(x)-\tilde{C}_{0}(x)\Big\|_{\mathcal{B}(\mathbb{C}^{2})}\le 2\kappa (1+|x|)^{-1-\epsilon_{0}},
\end{equation*}
where $\|\cdot\|_{\mathcal{B}(\mathbb{C}^{2})}$ is the operator norm on $\mathbb{C}^{2}$.
\end{Prop}
We introduce \lq\lq modified wave operators" as follows:
\begin{equation*}
W_{\pm}(U, U_{0}, J):=\text{s-}\displaystyle\lim_{t\rightarrow\pm\infty}U^{-t}JU_{0}^{t}\Pi_{\text{ac}}(U_{0}),
\end{equation*}
where $\Pi_{\text{ac}}(U_{0})$ is the orthogonal projection onto the absolutely continuous subspace of $U_{0}$.
\begin{Thm}\normalfont
$W_{\pm}(U, U_{0}, J)$ exist and are complete.
\end{Thm}
\begin{proof}
From Proposition 2.3, we can show that $C-\tilde{C}_{0}$ is trace-class [15, Lemma 2.1]. Thus $U-\tilde{U}_{0}$ is trace class. Then, it is seen that
\begin{equation*}
W_{\pm}(U, \tilde{U}_{0}):=\text{s-}\displaystyle\lim_{t\rightarrow\pm\infty}U^{-t}\tilde{U}^{t}_{0}\Pi_{\text{ac}}(\tilde{U}_{0})
\end{equation*}
exist and are complete ($\text{Ran}W_{\pm}=\mathcal{H}_{\text{ac}}(U)$) [15, Theorem 2.3]. Since $\tilde{U}^{t}_{0}=JU^{t}_{0}J^{-1}$ and $\Pi_{\text{ac}}(\tilde{U}_{0})=J\Pi_{\text{ac}}(U_{0})J^{-1}$, it is seen that 
\begin{equation*}
\begin{aligned}
\displaystyle\text{s-}\lim_{t\rightarrow\pm\infty}U^{-t}JU_{0}^{t}\Pi_{\text{ac}}(U_{0})&=\text{s-}\displaystyle\lim_{t\rightarrow\pm\infty}U^{-t}JU_{0}^{t}J^{-1}J\Pi_{\text{ac}}(U_{0})J^{-1}J
\\
&=\text{s-}\lim_{t\rightarrow\pm\infty}U^{-t}\tilde{U}_{0}^{t}\Pi_{\text{ac}}(\tilde{U}_{0})J
\\
&=W_{\pm}(U, \tilde{U}_{0})J.
\end{aligned}
\end{equation*}
This implies the existence of $W_{\pm}(U, U_{0}, J)$. Since $W_{\pm}(U, \tilde{U}_{0})$ are complete, we have $\text{Ran}(W_{\pm}(U, \tilde{U}_{0}))=\mathcal{H}_{\text{ac}}(U)$. Since $U_{0}$ has purely absolutely continuous spectrum (see Proposition  4.1 below), $J$ maps $\mathcal{H}_{\text{ac}}(U_{0})$ to $\mathcal{H}_{\text{ac}}(\tilde{U}_{0})$. Thus the completeness of $W_{\pm}(U, U_{0}, J)$ follows.
\end{proof}
\begin{Prop}\normalfont
It follows that
\begin{equation*}
\sigma_{\text{ess}}(U)=\sigma_{\text{ess}}(U_{0})=\begin{cases}
\{e^{it}|t\in[\delta/2+\zeta, \pi+\delta/2-\zeta]\cup[\pi+\delta/2+\zeta, 2\pi+\delta/2-\zeta]\},\hspace{3mm}\text{if }0<a<1,
\\
\mathbb{T}\hspace{102mm}\text{if }a=1.
 \end{cases}
\end{equation*}
\end{Prop}
\begin{proof}
From Proposition 2.3, $C-\tilde{C}_{0}$ is a compact operator. This implies that the compactness of $U-\tilde{U}_{0}=S(C-\tilde{C}_{0})$. By Lemma 2.2 of \cite{MS} and unitary invariance of essential spactrum, we have $\sigma_{\text{ess}}(U)=\sigma_{\text{ess}}(\tilde{U}_{0})=\sigma_{\text{ess}}(U_{0})$. The last equality follows from Proposition 2.1.
\end{proof}
\section{Commutator theory}
In this section, we recall some definitions and notations related to commutator theory. We mainly refer [2, 12]. We denote the set of bounded linear operators from a Hilbert space $\mathcal{H}_{0}$ to $\mathcal{H}$ by $\mathcal{B}(\mathcal{H}_{0},\mathcal{H})$ and $\mathcal{B}(\mathcal{H}):=\mathcal{B}(\mathcal{H}, \mathcal{H})$. Moreover, we denote the set of compact operators from $\mathcal{H}_{0}$ to $\mathcal{H}$ by $\mathcal{K}(\mathcal{H}_{0}, \mathcal{H})$ and $\mathcal{K}(\mathcal{H}):=\mathcal{K}(\mathcal{H}, \mathcal{H})$.

Let $T\in\mathcal{B}(\mathcal{H})$ and let $A$ be a self-adjoint operator on $\mathcal{H}$. We say that $T\in C^{k}(A)$ $(k\in\mathbb{N})$ if a $\mathcal{B}(\mathcal{H})$-valued map $\mathbb{R}\ni t\mapsto e^{-itA}Te^{itA}$ is belongs to $C^{k}$ class strongly. Especially in the case where $k=1$, it is known that $T\in C^{1}(A)$ if and only if a following form
\begin{equation*}
D(A)\ni \phi\mapsto \langle A\phi, T\phi\rangle-\langle \phi, TA\phi\rangle
\end{equation*}
can be continuously extended to the form on $\mathcal{H}$. We denote the operator correspond to continuous extension of the above form by $[A, T]$.

Here we introduce three regularity conditions which are stronger than $T\in C^{1}(A)$. $T\in C^{1, 1}(A)$ means that $T\in C^{1}(A)$ and 
\begin{equation*}
\displaystyle\int_{0}^{1}\|e^{-itA}Te^{itA}+e^{itA}Te^{-itA}-2S\|_{\mathcal{B}(\mathcal{H})}\displaystyle\frac{\text{d}t}{t^{2}}<\infty.
\end{equation*}
$T\in C^{1+0}(A)$ means that $T\in C^{1}(A)$ and 
\begin{equation*}
\displaystyle\int_{0}^{1}\|e^{-itA}[A, S]e^{itA}-[A, S]\|_{\mathcal{B}(\mathcal{H})}\displaystyle\frac{\text{d}t}{t}<\infty.
\end{equation*}
$T\in C^{1+\epsilon}$ for some $\epsilon>0$ means that $T\in C^{1}(A)$ and 
\begin{equation*}
\|e^{-itA}[A, S]e^{itA}-[A, S]\|_{\mathcal{B}(\mathcal{H})}\|\le \text{Const.}t^{\epsilon}\hspace{3mm}\text{for all $t\in(0, 1)$}
\end{equation*}
For above conditions, following inclusion relation holds [2, Section 5.2.4]:
\begin{equation*}
C^{2}(A)\subset C^{1+\epsilon}(A)\subset C^{1+0}(A)\subset C^{1, 1}(A)\subset C^{1}(A).
\end{equation*}

Next, we introduce two functions which are useful to consider the commutator theory for unitary operators which is introduced in \cite{RSTI}. For self-adjoint cases, see e.g. [2, Section 7.2]. We assume that $U\in C^{1}(A)$. For $T, S\in\mathcal{B}(\mathcal{H})$, we write $T\gtrsim S$ if there exists a compact operator $K\in\mathcal{K}(\mathcal{H})$ such that $T+K\ge S$. For $\theta\in\mathbb{T}$ and $\epsilon>0$, we set
\begin{equation*}
\Theta(\theta, \epsilon):=\{\theta'\in\mathbb{T}||\text{arg}(\theta-\theta')|<\epsilon\}, \hspace{5mm}E^{U}(\theta;\epsilon):=E^{U}(\Theta(\theta;\epsilon)).
\end{equation*}
where, $E^{U}(\cdot)$ is the spectral measure of $U$. Under above preparations, we introduce functions $\rho^{A}_{U}:\mathbb{T}\mapsto (-\infty, \infty]$ and $\tilde{\rho}^{A}_{U}:\mathbb{T}\mapsto (-\infty, \infty]$ by
\begin{equation*}
\rho^{A}_{U}(\theta):=\displaystyle\sup\{a\in\mathbb{R}|\exists \epsilon \text{ such that }E^{U}(\theta;\epsilon)U^{-1}[A, U]E^{U}(\theta;\epsilon)\ge aE^{U}(\theta;\epsilon)\},
\end{equation*}
and
\begin{equation*}
\tilde\rho^{A}_{U}(\theta):=\displaystyle\sup\{a\in\mathbb{R}|\exists \epsilon>0 \text{ such that }E^{U}(\theta;\epsilon)U^{-1}[A, U]E^{U}(\theta;\epsilon)\gtrsim aE^{U}(\theta;\epsilon)\}.
\end{equation*}
General facts related to commutator theory for unitary operators in one Hilbert space is considered in [12, Section 3.3]. The following fact is important to show the absence of singular continuous spectrum:
\begin{Thm}\normalfont [12, Theorem 3.6]
Let $U$ be a unitary operator and $A$ be a self-adjoint operator on $\mathcal{H}$. We assume either that $U$ has a spectral gap and $U\in C^{1, 1}(A)$ or $U\in C^{1+0}(A)$. Moreover, we also assume that there exists an open set $\Theta\subset\mathbb{T}$, $a>0$, and an operator $K\in\mathcal{K}(\mathcal{H})$ such that
\begin{equation*}
E^{U}(\Theta)U^{-1}[A, U]E^{U}(\Theta)\ge aE^{U}(\Theta)+K.
\end{equation*}
Then, $U$ has at most finitely many eigenvalues in $\Theta$, each one of finite multiplicity, and $U$ has no singular continuous spectrum in $\Theta$.
\end{Thm}
To show theorem 3.1,  in addition, we introduce the commutator theory in a two Hilbert space setting. We consider an another triple ($\mathcal{H}_{0}$, $U_{0}$, $A_{0}$) in addition to $(\mathcal{H}, U, A)$, where $\mathcal{H}_{0}$ is a Hilbert space, $U_{0}$ is a unitary operator on $\mathcal{H}$ and $A_{0}$ is a self-adjoint operator on $\mathcal{H}_{0}$. We also introduce a identification operator $J\in\mathcal{B}(\mathcal{H}_{0}, \mathcal{H})$. Following general result is important:
\begin{Thm}\normalfont [11, Theorem 3.7]
We assume that 
\begin{enumerate}
\item $U_{0}\in C^{1}(A_{0})$ and $U\in C^{1}(A)$, 
\item $JU^{-1}_{0}[A_{0}, U_{0}]J^{\ast}-U^{-1}[A, U]\in\mathcal{K}(\mathcal{H})$,
\item $JU_{0}-UJ\in\mathcal{K}(\mathcal{H}_{0}, \mathcal{H})$,
\item For each $f\in C(\mathbb{C},\mathbb{R})$, $f(U)(JJ^{\ast}-1)f(U)\in\mathcal{K}(\mathcal{H})$,
\end{enumerate}
Then, it follows that $\tilde{\rho}_{U}^{A}\ge \tilde{\rho}^{A_{0}}_{U_{0}}$
\end{Thm}
To apply the commutator theory for time evolution operator $U$ introduced in section 2, in what follows, we consider two triples $(\mathcal{H}, U, JA_{0}J^{\ast})$ and $(\mathcal{H}, U_{0}, A_{0})$. A following fact is useful to check the condition $U\in C^{1}(A)$ and the second condition in Theorem 4.2:
\begin{Thm}\normalfont [12, Corollary 3.11, Corollary 3.12]
Let $U_{0}\in C^{1}(A_{0})$. Suppose that $JA_{0}J^{\ast}$ is essentially self-adjoint on a set $\mathcal{D}$, and assume that
\begin{equation*}
\overline{BA_{0}\upharpoonright D(A_{0})}\in \mathcal{B}(\mathcal{H}), \hspace{3mm}\overline{B_{\ast}A_{0}\upharpoonright D(A_{0})}\in\mathcal{K}(\mathcal{H}), 
\end{equation*}
where $B:=JU_{0}-UJ_{0}$ and $B_{\ast}:=JU_{0}^{\ast}-U^{\ast}J$. Then, $U\in C^{1}(JA_{0}J^{\ast})$ and $JU_{0}^{-1}[A_{0}, U_{0}]J^{\ast}-U^{-1}[JA_{0}J^{\ast}, U]\in\mathcal{K}(\mathcal{H})$.
\end{Thm}

\section{Spectral analysis for quantum walks}
In this section, we show the absence of singular continuous spectrum of $U$. First, we introduce the asymptotic velocity operator of $U_{0}=SC_{0}$ by
\begin{equation*}
\widehat{V_{0}\psi}(k)=\displaystyle\sum_{j=1, 2}\displaystyle\frac{i\lambda_{j}(k)}{\lambda_{j}(k)}\langle u_{j}(k), \hat{\psi}(k)\rangle_{\mathbb{C}^{2}}u_{j}(k),\hspace{3mm}x\in [0, 2\pi),\hspace{3mm}\psi\in \mathcal{H}.
\end{equation*}
Note that $V_{0}$ is bounded and self-adjoint on $\mathcal{H}$. 

For any $\psi, \phi\in C([0, 2\pi), \mathbb{C}^{2})$, we introduce the operator $|\psi\rangle\langle\phi|:C([0, 2\pi), \mathbb{C}^{2})\rightarrow C([0,2\pi), \mathbb{C}^{2})$ by
\begin{equation*}
\big(|\psi\rangle\langle\phi|f\big)(k):=\langle \psi(k), f(k)\rangle_{\mathbb{C}^{2}} \phi(k),\hspace{3mm}f\in C([0, 2\pi, \mathbb{C}^{2}),\hspace{3mm}k\in[0, 2\pi)
\end{equation*}
This operator can be continuously extended to a bounded operator on $\mathcal{H}$. Moreover, we introduce the self-adjoint operator $P$ in $\mathcal{K}$ as follows:
\begin{equation*}
\begin{aligned}
D(P)&:=\{f\in \mathcal{K}|f\text{ is absolutely continuous }, f'\in\mathcal{K}, \text{ and }f(0)=f(2\pi)\},
\\
(Pf)&:=-if',\hspace{3mm}f\in D(P).
\end{aligned}
\end{equation*}
Under above notations, we introduce the operator $X$ by
\begin{equation*}
\widehat{X}f(k):=-\displaystyle\sum_{j=1, 2}(|u_{j}\rangle\langle u_{j}|P-i|u_{j}\rangle\langle u'_{j}|)f,\hspace{3mm}f\in\mathcal{F}\mathcal{H}_{\text{fin}}.
\end{equation*}
X is essentially self-adjoint [12, Lemma 4.3] and we denote the closure of $X$ by the same symbol. Moreover we introduce the following operator:
\begin{equation*}
A_{0}:=\displaystyle\frac{1}{2}(XV_{0}+V_{0}X).
\end{equation*}
$A_{0}$ is self-adjoint and essentially self-adjoint on $\mathcal{H}_{\text{fin}}$.
\begin{Prop}\normalfont [12, Proposition 4.5]
Following properties hold:
\begin{enumerate}
\item $U_{0}\in C^{1}(A_{0})$ and $U_{0}^{-1}[A_{0}, U_{0}]=V_{0}^{2}$.
\item $\rho^{A_{0}}_{U_{0}}=\tilde{\rho}^{A_{0}}_{U_{0}}$ and
\begin{enumerate}
\item if $a\in (0, 1)$, then $\tilde{\rho}^{A_{0}}_{U_{0}}(\theta)>0$ for $\theta\in \text{Int}(\sigma(U_{0}))$, $\tilde{\rho}^{A_{0}}_{U_{0}}(\theta)=0$ for $\theta\in\partial\sigma(U_{0})$, and $\tilde{\rho}^{A_{0}}_{U_{0}}(\theta)=\infty$ otherwise, 
\item if $a=1$, then $\tilde{\rho}^{A_{0}}_{U_{0}}(\theta)=1$ for all $\theta\in\mathbb{T}$.
\end{enumerate} 
\item If $a\in (0, 1)$, then $U_{0}$ has purely absolutely continuous spectrum and
\begin{equation*}
\sigma(U_{0})=\sigma_{\text{ac}}(U_{0})=\{e^{i\gamma}|\gamma\in[\delta/2+\zeta, \pi+\delta/2-\zeta]\cup [\pi+\delta/2+\zeta, 2\pi+\delta/2-\zeta]\}
\end{equation*}
\item If $a=1$, then $U_{0}$ has purely absolutely continuous spectrum and $\sigma(U_{0})=\sigma_{\text{ac}}(U_{0})=\mathbb{T}$.
\end{enumerate}
\end{Prop}
In what follows, we set $A:=JA_{0}J^{\ast}$. We show conditions in Theorem 3.2 for two triples $(\mathcal{H}, U, A)$ and $(\mathcal{H}, U_{0}, A_{0})$.
\begin{Lem}\normalfont
It follows that $U\in C^{1}(A)$ and $JU_{0}^{-1}[A_{0}, U_{0}]J^{\ast}-U^{-1}[A, U]\in\mathcal{K}(\mathcal{H})$.
\end{Lem}
\begin{proof}
From Proposition 4.1, we know $U_{0}\in C^{1}(A_{0})$. Moreover, $JA_{0}J^{-1}$ is essentially self-adjoint on $\mathcal{D}=\mathcal{H}_{\text{fin}}$ since $J$ is unitary and $J\mathcal{H}_{\text{fin}}=\mathcal{H}_{\text{fin}}$. Now we check two conditions in Theorem 4.3. We note that $A_{0}$ has a following form on $\mathcal{H}_{\text{fin}}$:
\begin{equation*}
A_{0}=QK+\displaystyle\frac{i}{2}H_{0}
\end{equation*}
for some $K, H_{0}\in\mathcal{B}(\mathcal{H})$, where $Q$ is the position operator defined by
\begin{equation*}
D(Q):=\{\psi\in\mathcal{H}|\displaystyle\sum_{x\in\mathbb{Z}}x^{2}\|\psi(x)\|_{\mathbb{C}^{2}}<\infty\},\hspace{2mm}
(Q\psi)(x):=x\psi(x),\hspace{3mm}x\in\mathbb{Z}.
\end{equation*}
For more details, see the proof of [12, Lemma 4.10]. On $\mathcal{H}_{\text{fin}}$, it follows that
\begin{equation*}
\begin{aligned}
BA_{0}=(JU_{0}J^{\ast}-U)J(QK+\displaystyle\frac{i}{2}H_{0})
&=(\tilde{U}_{0}-U)QJK+\displaystyle\frac{i}{2}(\tilde{U}_{0}-U)H_{0}
\\
&=S(\tilde{C}_{0}-C)QJK+\displaystyle\frac{i}{2}(\tilde{U}_{0}-U)H_{0},
\end{aligned}
\end{equation*}
where we used the commutativity of $J$ and $Q$. From Proposition 2.3, we see that $\tilde{C}_{0}-C$ is a compact operator and $(\tilde{C}_{0}-C)Q$ can be extended to a compact operator on $\mathcal{H}$. Thus we have $\overline{BA_{0}\upharpoonright D(A_{0})}\in \mathcal{K}(\mathcal{H})\subset \mathcal{B}(\mathcal{H})$. By the similar manner, it follows on $\mathcal{H}_{\text{fin}}$ that 
\begin{equation*}
\begin{aligned}
B_{\ast}A_{0}&=(JU^{\ast}_{0}J^{\ast}-U^{\ast})J(QK+\displaystyle\frac{i}{2}H_{0})
\\
&=(\tilde{C}_{0}^{\ast}-C^{\ast})SJQ\mathcal{F}^{-1}K\mathcal{F}+\displaystyle\frac{i}{2}(\tilde{U}_{0}-U)^{\ast}JH_{0}
\\
&=(\tilde{C}_{0}-C)^{\ast}QSJK+(\tilde{C}_{0}-C)^{\ast}(SQ-QS)JK+\displaystyle\frac{i}{2}(\tilde{U}_{0}-U)^{\ast}JH_{0}
\end{aligned}
\end{equation*}
Since $(\tilde{C}_{0}-C)^{\ast}$ and $(\tilde{U}_{0}-U)^{\ast}$ are compact, $(\tilde{C}_{0}-C)^{\ast}Q$ can be extended to a compact operator on $\mathcal{H}$ and $SQ-QS$ can be extended to a bounded operator on $\mathcal{H}$, we have $\overline{B_{\ast}A_{0}\upharpoonright D(A_{0})}\in\mathcal{K}(\mathcal{H})$. An application of Theorem 3.3 implies the desired result.
\end{proof}
Since $J$ is unitary, $JJ^{\ast}=1$ holds. Moreover, $JU_{0}-UJ=(JU_{0}J-U)J=(\tilde{U}_{0}-U)J\in\mathcal{K}(\mathcal{H})$ since $\tilde{U}_{0}-U$ is compact. Therefore, we checked conditions in Theorem 4.2. 
 We introduce the set of threshold of $U$ by $\tau(U):=\partial\sigma(U_{0})$,
where $\partial\sigma(U_{0})$ is the set of boundary of $\sigma(U_{0})$ in $\mathbb{T}$. We note that $\tau(U)$ contains at most 4 values. 
\begin{Prop}\normalfont
We have $\tilde{\rho}^{A}_{U}\ge \tilde{\rho}^{A_{0}}_{U_{0}}$. In particular, if $\theta\in\sigma(U_{0})\setminus \tau(U)$, then $\tilde{\rho}^{A_{0}}_{U_{0}}(\theta)>0$. 
\end{Prop}
\begin{proof}
$\tilde{\rho}^{A}_{U}\ge \tilde{\rho}^{A_{0}}_{U_{0}}$ follows by an application of Theorem 3.2. The latter assertion follows from Proposition 4.1.
\end{proof}
To apply Theorem 3.1, we have to check a regularity of $U$ more detail.
\begin{Lem}\normalfont
For any $\epsilon\in (0, 1)$ with $\epsilon\le \epsilon_{0}$, $U\in C^{1+\epsilon_{0}}(A)$. Here $\epsilon_{0}>0$ is a constant introduced in Assumption 2.1.
\end{Lem}
\begin{proof}
This proof is a slight modification of [12, Lemma 4.13]. In the proof of Proposition 4.5 of [12], we see that $U_{0}\in C^{2}(A_{0})$. Since $J$ is unitary, it follows that $\tilde{U}_{0}\in C^{2}(A)\subset C^{1+\epsilon}(A)$. We decompose $U$ as $U=\tilde{U}_{0}+(U-\tilde{U}_{0})$. Thus it suffices to show that $U-\tilde{U}_{0}\in C^{1+\epsilon}(A)$. We see that 
\begin{equation*}
D_{0}:=A(U-\tilde{U}_{0})-(U-\tilde{U}_{0})A
\end{equation*}
on $\mathcal{H}_{\text{fin}}$ can be extended to a bounded operator on $\mathcal{H}$. We denote it by the same symbol. According to [2, p.325-328] or [12, Lemma 4.13], following estimate holds:
\begin{equation*}
\begin{aligned}
\|e^{-itA}D_{0}e^{itA}-D_{0}\|_{\mathcal{B}(\mathcal{H})}
&\le \text{Const.}(\|\sin(tA)D_{0}\|_{\mathcal{B}(\mathcal{H})}+\sin(tA)D_{0}^{\ast}\|_{\mathcal{B}(\mathcal{H})})
\\
&\le \text{Const.}(\|tA(tA+i)^{-1}D_{0}\|_{\mathcal{B}(\mathcal{H})}+\|tA(tA+i)^{-1}D_{0}^{\ast}\|_{\mathcal{B}(\mathcal{H})})
\end{aligned}
\end{equation*}
We set $A_{t}:=tA(tA+i)^{-1}$ and $\Lambda_{t}:=t\langle Q\rangle(\langle Q\rangle+i)^{-1}$ with $\langle Q\rangle:=\sqrt{Q^{2}+1}$. We note that $A\langle Q\rangle^{-1}\in\mathcal{B}(\mathcal{H})$. Then it follows that 
\begin{equation*}
A_{t}=(A_{t}+i(tA+i)^{-1}A\langle Q\rangle^{-1})\Lambda_{t}.
\end{equation*}
Since $A_{t}+i(tA+i)^{-1}A\langle Q\rangle^{-1}$ is bounded, it suffices to show that 
\begin{equation*}
\|\Lambda_{t}D_{0}\|_{\mathcal{B}(\mathcal{H})}+\|\Lambda_{t}D_{0}^{\ast}\|_{\mathcal{B}(\mathcal{H})}\le \text{Const. }t^{\epsilon}\hspace{5mm}t\in(0, 1).
\end{equation*}
We have to show that operators $\langle Q\rangle^{\epsilon}D_{0}$ and $\langle Q\rangle^{\epsilon}D_{0}$ defined on the form sense on $\mathcal{H}_{\text{fin}}$ extended to a bounded operator on $\mathcal{H}$. We note that $\langle Q\rangle^{1+\epsilon}(C-\tilde{C}_{0})\in\mathcal{B}(\mathcal{H})$ and $\langle Q\rangle^{-1}A_{0}$ defined in the form sense on $\mathcal{H}_{\text{fin}}$ extend to a bounded operator on $\mathcal{H}$. This implies that $\langle Q\rangle^{\epsilon}D_{0}$ and $\langle Q\rangle^{\epsilon}D_{0}^{\ast}$ defined in the form sense on $\mathcal{H}_{\text{fin}}$ extend to bounded operators on $\mathcal{H}$. Thus the proof is completed.
\end{proof}
By Theorem 3.1, Proposition 4.2 and Lemma 4.2, we have the following result.
\begin{Thm}\normalfont
For any closed set $\Theta\subset\mathbb{T}\setminus\tau(U)$, the operator $U$ has at most finitely many eigenvalues in $\Theta$, each one of finite multiplicity, and $U$ has no singular continuous spectrum in $\Theta$.
\end{Thm}
Recall that $\tau(U)$ is a finite set. From Theorem 4.1, $U$ have no singular continuous spectrum.

\section{Derivation of weak limit theorem}
We set $Q_{0}(t):=U_{0}^{-t}QU_{0}^{t}$.
\begin{Thm}\normalfont [15, Theorem 4.1] 
It follows that
\begin{equation*}
\displaystyle\text{s-}\lim_{t\rightarrow\infty}e^{i\xi Q_{0}(t)}=e^{i\xi V_{0}},\hspace{3mm}\xi\in\mathbb{R}.
\end{equation*}
\end{Thm}
Let $X_{t}$ be a random variable which describes the position of a quantum walker with $U$ and an initial state $\Psi_{0}$ at time $t\in\mathbb{Z}$. The probability distribution of $X_{t}$ is given by
\begin{equation*}
\mathbb{P}(\{X_{t}=x\})=\|(U^{t}\Psi_{0})(x)\|_{\mathbb{C}^{2}}^{2}, \hspace{5mm}x\in\mathbb{Z}.
\end{equation*}
Moreover, we also introduce the characteristic function of the average velocity $X_{t}/t$ of a quantum walker by
\begin{equation*}
\mathbb{E}[e^{i\xi X_{t}/t}]:=\langle \Psi_{0}, e^{itQ(t)/t}\Psi_{0}\rangle,\hspace{5mm}\xi\in\mathbb{R},
\end{equation*}
where $Q(t):=U^{-t}QU^{t}$. Our interest is the limit of $X_{t}/t$ in a weak sense. 
\begin{Thm}\normalfont
We set $V^{+}_{J}:=W_{+}(U, U_{0}, J)V_{0}W_{+}(U, U_{0}, J)^{\ast}$. Then for any $\xi\in\mathbb{R}$, it follows that
\begin{equation*}
\text{s-}\displaystyle\lim_{t\rightarrow\infty}e^{i\xi Q(t)/t}=\Pi_{\text{p}}(U)+e^{i\xi V_{J}^{+}}\Pi_{\text{ac}}(U),
\end{equation*}
where $\Pi_{\text{p}}(U)$ is the orthogonal projection onto a subspace generated by eigenvectors of $U$.
\begin{proof}
Since $U$ have no continuous spectrum, we can decompose that
\begin{equation*}
\text{s-}\lim_{t\rightarrow\infty}e^{i\xi Q(t)/t}=\text{s-}\lim_{t\rightarrow\infty}\big(e^{i\xi Q(t)/t}\Pi_{\text{p}}(U)+e^{i\xi Q(t)/t}\Pi_{\text{ac}}(U)\big).
\end{equation*}
By [15, Theorem 4.2], we have $\text{s-}\lim_{t\rightarrow\infty}e^{i\xi Q(t)/t}\Pi_{\text{p}}(U)=\Pi_{\text{p}}(U)$. For the absolutely continuous part, we consider the following decomposition:
\begin{equation*}
\begin{aligned}
&\hspace{5mm}e^{i\xi Q(t)/t}\Pi_{\text{ac}}(U)-e^{i\xi V_{J}^{+}}\Pi_{\text{ac}}(U)
\\
&=U^{-t}e^{i\xi Q/t}U^{t}\Pi_{\text{ac}}(U)-W_{+}(U, U_{0}, J)e^{i\xi V_{0}}W_{+}(U, U_{0}, J)^{\ast}\Pi_{\text{ac}}(U)
\\
&=U^{-t}JU_{0}^{t}(U_{0}^{-t}e^{i\xi Q/t}U_{0}^{t})U_{0}^{-t}J^{-1}U^{t}\Pi_{\text{ac}}(U)-W_{+}(U,U_{0}, J)e^{i\xi V_{0}}W_{+}(U, U_{0}, J)^{\ast}\Pi_{\text{ac}}(U)
\\
&=U^{-t}JU_{0}^{t}e^{i\xi Q_{0}(t)/t}\big(U_{0}^{-t}J^{-1}U^{t}\Pi_{\text{ac}}(U)-W_{+}(U, U_{0}, J)^{\ast}\big)\Pi_{\text{ac}}(U)
\\
&\hspace{5mm}+U^{-t}JU_{0}^{t}\big(e^{i\xi Q(t)_{0}/t}-e^{i\xi V_{0}}\big)W_{+}(U, U_{0}, J)^{\ast}\Pi_{\text{ac}}(U)
\\
&\hspace{5mm}+\big(U^{-t}JU_{0}^{t}-W_{+}(U, U_{0}, J)\big)e^{i\xi V_{0}}W_{+}(U, U_{0}, J)^{\ast}\Pi_{\text{ac}}(U),
\end{aligned}
\end{equation*}
where we used the strong commutativity of $Q$ and $J$. We note that $W_{+}(U, U_{0}, J)^{\ast}$ maps $\mathcal{H}_{\text{ac}}(U)$ to $\mathcal{H}_{\text{ac}}(U_{0})$ and $V_{0}$ leaves $\mathcal{H}_{\text{ac}}(U_{0})$ invariant. By Theorem 5.1, it is seen that $\text{s-}\lim_{t\rightarrow\infty}e^{i\xi Q_{0}/t}=e^{i\xi V_{0}}$. By taking a limit $t\rightarrow\infty$,  the desired result follows. 
\end{proof}
\end{Thm}

\begin{Thm}\normalfont
Let $\Psi_{0}\in\mathcal{H}$ be an initial state with $\|\Psi_{0}\|=1$ and $V$ be the random variable whose probability distribution is given by
\begin{equation*}
\mu_{V}(dv):=\|\Pi_{\text{p}}(U)\Psi_{0}\|^{2}\delta_{0}dv+\|E_{V_{J}}^{+}(\cdot)\Pi_{\text{ac}}(U)\Psi_{0}\|^{2}dv,
\end{equation*}
where $\delta_{0}$ is the Dirac measure for the point 0 and $E_{V_{J}^{+}}(\cdot)$ is the spectral measure of $V_{J}^{+}$.
Then it follows that
\begin{equation*}
\displaystyle\lim_{t\rightarrow\infty}\mathbb{E}[e^{i\xi Q(t)/t}]=\mathbb{E}[e^{i\xi V}],\hspace{5mm}\xi\in\mathbb{R}.
\end{equation*}
\end{Thm}
\begin{proof}
The proof is quite similar to [15, Corollary 2.4]. We omit the proof.
\end{proof}
\vspace{5mm}
\textbf{Acknowledgments}
The author would like to thank A. Suzuki for various comments and constant encouragements. The author would also like to thank H. Ohno and S. Richard for helpful comments. This work was supported by the Research Institute of Mathematical Sciences, an International Joint Usage/Research center located in Kyoto university.

\end{document}